%% file: main.tex
\documentclass[a4paper,onecolumn,11pt,accepted=2023-06-27]{quantumarticle}

\pdfoutput=1

\usepackage[utf8]{inputenc}
\usepackage[english]{babel}
\usepackage[T1]{fontenc}

\usepackage{hyperref}
\hypersetup{breaklinks=true}

\PassOptionsToPackage{usenames, dvipsnames}{color}

\input{libs/graphics}
\input{libs/math}
\input{libs/complexity}

\input{libs/quantum}
\input{libs/editorial}

\currentrevision{2}

\usetikzlibrary{shapes.geometric, arrows,arrows.meta, positioning}

\begin{document}

\title{On hitting times for general quantum Markov processes}

\author{Lorenzo Laneve}
\email{lorenzo.laneve@usi.ch}
\affiliation{Faculty of Informatics — Universit\`a della Svizzera Italiana, 6900 Lugano, Switzerland}
\affiliation{IBM Quantum, IBM Research — Zurich, S\"aumerstrasse 4, 8803 R\"uschlikon, Switzerland}

\author{Francesco Tacchino}
\email{fta@zurich.ibm.com}
\affiliation{IBM Quantum, IBM Research — Zurich, S\"aumerstrasse 4, 8803 R\"uschlikon, Switzerland}

\author{Ivano Tavernelli}
\email{ita@zurich.ibm.com}
\affiliation{IBM Quantum, IBM Research — Zurich, S\"aumerstrasse 4, 8803 R\"uschlikon, Switzerland}

\begin{abstract}
Random walks (or Markov chains) are models extensively used in theoretical computer science. Several tools, including analysis of quantities such as hitting and mixing times, are helpful for devising randomized algorithms. A notable example is Sch\"oning's algorithm for the satisfiability (SAT) problem. In this work, we use the density-matrix formalism to define a \emph{quantum Markov chain} model which directly generalizes classical walks, and we show that a common tool such as hitting times can be computed \revdtext{analytically} with a formula similar to the one found in the classical theory, which we then apply to known quantum settings such as Grover's algorithm.
\end{abstract}

\input{contents/introduction}

\input{contents/markov-chains}
\input{contents/hitting-times}
\input{contents/generalized-ht}
\input{contents/grover}
\input{contents/cycle}
\input{contents/qiskit}
\input{contents/conclusions}

\section*{Acknowledgements}
We thank David Sutter, Sabrina Maniscalco and Stefan Wolf for insightful discussions and feedback. We acknowledge the use of IBM Quantum services for this work. IBM, the IBM logo, and ibm.com are trademarks of International Business Machines Corp., registered in many jurisdictions worldwide. Other product and service names might be trademarks of IBM or other companies. 
The current list of IBM trademarks is available at \url{https://www.ibm.com/legal/copytrade}. LL acknowledges support from the Swiss National Science Foundation (SNF), grant No.~\texttt{200020\_214808}.

\bibliographystyle{quantum}
\bibliography{refs}

\end{document}

%% file: libs/graphics.tex
\usepackage{graphicx}
\usepackage{subcaption}
\usepackage{float}

\usepackage{tikz}
\usepackage{pgfplots}
\pgfplotsset{compat=1.18}

\usetikzlibrary{angles,quotes}
\usetikzlibrary{decorations.pathmorphing,arrows.meta}

%% file: libs/math.tex
\usepackage{bbm}
\usepackage{amsmath}
\usepackage{amssymb}
\usepackage{amsthm}
\usepackage{mathtools}

\providecommand\given{}
\newcommand\SetSymbol[1][]{%
  \nonscript\:#1\vert
  \allowbreak
  \nonscript\:
  \mathopen{}
}

\DeclarePairedDelimiterX\pSet[1]{\{}{\}}{%
  \renewcommand\given{\SetSymbol[\delimsize]}
  #1
}

\ifdefined\skiptheoremsetup\else
    \newtheorem{theorem}{Theorem}

    \newtheorem{lemma}[theorem]{Lemma}

    \newtheorem{definition}{Definition}

\fi

\newcommand{\N}{\mathbb{N}}

\newcommand{\Z}{\mathbb{Z}}

\newcommand{\calE}{\mathcal{E}}

\newcommand{\calG}{\mathcal{G}}
\newcommand{\calH}{\mathcal{H}}
\newcommand{\calI}{\mathcal{I}}

\newcommand{\calP}{\mathcal{P}}

\newcommand{\calU}{\mathcal{U}}

\newcommand{\vv}[1]{\mathbf{#1}}

%% file: libs/complexity.tex
\usepackage{xparse}

\newcommand{\bigO}{\mathcal{O}}

\providecommand\given{}

\DeclarePairedDelimiterX\pBrackets[1]{[}{]}{%
\renewcommand\given{\SetSymbol[\delimsize]}
#1
}

\newcommand{\Prob}[2][]{
 \def\tmp{#1}%
   \ifx\tmp{}
     \operatorname{Pr}\pBrackets*{#2}
   \else
     \operatorname{Pr}_{#1}\pBrackets*{#2}
   \fi
}

\newcommand{\Expect}[2][]{
 \def\tmp{#1}%
   \ifx\tmp{}
     \operatorname{\mathbb{E}}\pBrackets*{#2}
   \else
     \operatorname{\mathbb{E}}_{#1}\pBrackets*{#2}
   \fi
}

%% file: libs/quantum.tex
\usepackage{physics}
\usepackage{braket}

\usetikzlibrary{quantikz}

\renewcommand{\ket}[1]{| #1 \rangle}
\renewcommand{\bra}[1]{\langle #1 |}

\renewcommand{\proj}[1]{\ket{#1}\bra{#1}}
\renewcommand{\ketbra}[2]{\ket{#1}\bra{#2}}

\newcommand{\vspan}[1]{\operatorname{span}\{ #1 \}}

\renewcommand{\hom}{\mathrm{Hom}}

\newcommand{\id}{\mathbbm{1}}
\renewcommand{\tr}{\operatorname{tr}}

\newcommand{\sL}[2]{{\mathfrak{s l}(2, \mathbb{C})}}

%% file: libs/editorial.tex
\usepackage{ifthen}

\newcommand\curRevision{1}
\newcommand{\currentrevision}[1]{
    \def\curRevision{#1}
}

\newcommand{\revdtext}[2][1]{\ifthenelse{\equal{#1}{\curRevision}}{\textcolor{blue}{#2}}{#2}}

%% file: contents/introduction.tex
\section{Introduction}
Random walks (or Markov chains) are stochastic models that are extensively used in theoretical computer science. \revdtext{Classically, a random walk is defined through a graph, where the nodes are the possible states of the process and edges represent possible transitions. At each step, an outgoing edge from the current state is chosen according to some probability distribution, and the corresponding state is reached. A theory of Markov chains underlies the analysis of many algorithms:} a notable example is Sch\"oning's algorithm, one of the best known classical algorithms for the satisfiability (SAT) problem~\cite{schoning99}. \revdtext{An important property of Markov chains is the so-called \emph{hitting time}, which quantifies how many steps of the walk (in expectation) we need to execute in order to reach, or hit some fixed target state, given some initial conditions.} The analysis of hitting times is a powerful tool in search problems~\cite{childsgoldstone1, kroviadiabatic, magniezsearch, chakrabortysearch}, as these quantities are usually tightly related to complexity metrics. \revdtext{As an example, consider the satisfiability problem: given a formula $F(x)$, we start from some assignment $x^0$ (e.g., $x^0 = (0, \ldots, 0)$) and, at each step, we choose a variable to flip uniformly at random. This can be formalized as a random walk on an hypercube and, given a satisfying assignment $x^*$ for $F$, the hitting time from $x^0$ to $x^*$ tells us the number of steps needed on average to reach such assignment. An algorithm running this Markov chain and checking at each step whether the current state satisfies $F$ has time complexity that is proportional to the hitting time.}

\revdtext{Over the past decades, several research efforts were devoted to extending the concept of random walk to a quantum setting, with the aim of possibly achieving certain speedups leveraging quantum effects, such as constructive/destructive interference.} Notable examples are the works by Kempe~\cite{kempeintro, kempehittingtimes}, where the \emph{unitary walk model} is introduced: such a walk consists in repeatedly applying a unitary transformation to a pure (finite-dimensional) quantum state
\begin{align*}
    \ket{\phi_{t+1}} = U \ket{\phi_t}
\end{align*}
\revdtext{where the elements of the computational basis can be seen as vertices of a graph.} While many interesting results have been proven over the years, including novel algorithms~\cite{ambainisdistinctness, subsetfinding, kempehittingtimes, kempesearch, layden2022, mazzolasampingrates} and general constructions for quadratic speed-ups~\cite{aharonovmixing, szegedymixingspeedup, tulsi2008, ambainisquadraticspeedup2019}, the unitary walk model cannot by itself be considered a full generalization of classical random walks, as unitary evolution is purely reversible. Moreover, effects such as decoherence and classical randomness cannot be analyzed as part of the behaviour of these walks. \revdtext{Characterising the effect noise is not only relevant for implementations: in fact, Kendon and Tregenna~\cite{kendonreview, kendontregennanumerical} presented a series of numerical simulations in different settings, showing that decoherence can play a central role in providing additional speed-up.} \revdtext{Thus, in order to construct a complete theory that unifies unitary walks and classical Markov chains, it is necessary to consider open-system dynamics. Among these models, which use density matrices and completely positive (CP) maps~\cite{nielsenchuang, budini2004}, one can find the transition operation matrix (TOM) formalism~\cite{guddertomtem} and quantum stochastic walks~\cite{qswframework}. 
The interested reader can find more comprehensive reviews of quantum walk models in~\cite{kadianreview, venegasreview}, where several definitions and applications are listed.}

\revdtext{Even within a generic quantum setting, defining quantum hitting times is not straightforward. One has for instance to consider that checking through a quantum measurement whether the target state was hit will in general perturb the walk process itself. 
As a consequence, different notions of quantum hitting times have been formalized, depending on how they overcome the measurement problem~\cite{kempehittingtimes,szegedymixingspeedup,magniezhitting,Atia_2021,ambainis2004coins}. The definitions proposed in the literature can essentially be grouped into two classes. The first one is a \emph{one-shot} version, based on the first point in time where the overlap with the target state reaches some predefined threshold. The second is instead a so-called \emph{concurrent} definition, in which one checks at each step whether the target state was hit (a measurement whose observable has only two distinct eigenvalues). 
Clearly, these two formal definitions are diametrically opposed, in the sense that the former only requires one full final measurement at the end of the process, thus fully preserving the coherent behaviour of the walk, while the latter carries out one-bit measurements in between each execution of the walk.} Although the one-shot definitions are perfectly exploitable in algorithms, they are harder to analyze. In particular, since they are far from the usual classical definition, one cannot directly apply \revdtext{analytical} results known for classical hitting times. \revdtext{On the other hand, the concurrent formalization can be more clearly interpreted as a direct generalization of the classical definition, but require a specific analysis of the robustness of potential quantum speedups with respect to frequent measurement processes~\cite{barkaimeasurementinduced,barkaipassagetime,tornowmeasurementinduced}.}

\revdtext{In this work, we first introduce (Section~\ref{sec:quantum-markov-chains}) a discrete version of the density matrix-based quantum stochastic walk framework~\cite{qswframework} and show that not only it represents the most natural generalisation of the classical walk model, but that it also allows us to easily describe the effect of noise and measurements as quantum maps. Based on this formalism, we then propose (Sections~\ref{sec:hitting-times} and~\ref{sec:generalized-hitting-times}) an extension of the notion of concurrent hitting time, originally given by Kempe~\cite[Def.~3.3]{kempehittingtimes} only for the case of unitary walks, which circumvents the measurement problem by randomising the times at which the process is observed and reusing post-measurement states. We leverage the algebra of completely positive (CP) maps to derive a clean analytical formula to compute such hitting time, which we find to be completely analogous and consistent with its classical counterpart. Finally, we demonstrate some concrete applications of our results to some specific examples (Sections~\ref{sec:grover-algorithm} and~\ref{sec:hitting-times-cycle}), including quantum walks on cyclic graphs and a version of Grover's algorithm.}

%% file: contents/markov-chains.tex
\section{Quantum Markov Chains}
\label{sec:quantum-markov-chains}
In classical theory, the notion of random walk is formalized using a stochastic process called \emph{Markov chain}.
\begin{definition}
    A Markov chain is a tuple $(S, P)$ where
    \begin{itemize}
        \item $S$ is a finite set of states;
        \item $P \in [0,1]^{|S|\times|S|}$ is a stochastic matrix, i.e., a entry-wise non-negative matrix where each row sums up to $1$.
    \end{itemize}
    The matrix $P$ defines how a general probability distribution of states evolves at the next step of the process.
\end{definition}
The system is in a state $X_t \in S$ at time $t \in \N$, where $X_t$ follows a distribution $q_t$, which can be seen as a (row) vector in $S$, whose entries are non-negative and sum up to $1$ (the so-called \emph{state probability vector}). The stochastic matrix $P$ is used to retrieve the distribution $q_{t+1}$ of the state in the next time step:
\begin{align*}
    q_{t+1} = q_t P \Longrightarrow q_t = q_0 P^t .
\end{align*}
Therefore, a classical random walk can be seen as a (linear) evolution of a state $q$. A more detailed introduction to the topic can be found, e.g., in~\cite{norrismarkovchains}.

In order to analyze quantum random walks, we need a suitable model that conveniently defines such processes. 
As argued in the introduction, the unitary model is purely reversible and thus it cannot simulate a classical Markov chain (unless the stochastic matrix $P$ of the chain is a permutation of the standard basis vectors and every probability distribution is trivial). Moreover, works by Kendon and Tregenna~\cite{kendonreview, kendontregennanumerical} involving numerical experiments suggest that classical randomness and noise may even accelerate the behaviour of quantum walks in some cases, and such feature cannot be expressed in the unitary model.

In this work we aim at using a more general model to formalize the process using \revdtext{completely positive (CP)} maps and density matrices. This approach can also be seen as the discrete-time counterpart of the quantum stochastic walk~\cite{qswframework}.
\begin{definition}
    A quantum Markov chain consists of a tuple $(\calH, \calE)$ where:
    \begin{itemize}
        \item $\calH$ is a finite-dimensional Hilbert space;
        \item $\calE: \hom(\calH) \rightarrow \hom(\calH)$ is a completely positive trace-preserving map.
    \end{itemize}
    \revdtext{Here $\hom(\calH)$ represents the set of homomorphism of $\calH$ onto itself, i.e., the set of the corresponding square matrices.}
    The evolution is defined through a repeated application of a CP map $\calE$ onto a (generally mixed) quantum state in~$\calH$.
\end{definition}
We show that this formalization is a straightforward extension of the classical Markov model: indeed, the density matrix formalism can be seen as a state probability vector (the diagonal of the matrix) completed with quantum coherences. Therefore, a general quantum random walk is simply a (linear) transformation of a density matrix:
\begin{align*}
    \rho_{t+1} = \calE(\rho_t) \Longrightarrow \rho_t = \calE^t(\rho_0) .
\end{align*}
This model has several advantages over the unitary formalization: quantum channels can represent any hybrid quantum-classical process. In particular, an arbitrary Markov chain $(S, P)$ can be simulated by a quantum Markov chain $(\calH, \calE)$ where $\calH = \vspan{\ket{x} : x \in S}$ and
\begin{align}
    \label{eq:quantum-to-classical-construction}
    \calE(\rho) = \sum_{x,y \in S} [P]_{xy} \ketbra{y}{x} \rho \ketbra{x}{y} .
\end{align}
This construction projects all our quantum states onto a fixed measurement basis (in this case, the computational basis), and any quantum coherence is ignored, leaving us with a classical Markovian dynamics.
On the other hand, also unitary walks are particular cases of this model, as the application of a unitary can always be represented through a quantum channel. Therefore, this model unifies the two theories, while its expressiveness allows to include classical random choices and decoherence into the picture. Moreover, there could be some type of quantum walk that lies in the middle between fully coherent and fully classical processes. For example, Kendon et al.~\cite{kendontregennanumerical} already showed that such hybrid dynamics could be beneficial in some contexts.

%% file: contents/hitting-times.tex
\section{Hitting Times}
\label{sec:hitting-times}
Using the formalism introduced above, it is possible to extend classical results to the quantum setting by analogy. For example, if we translate a definition of hitting times proposed by Kempe (in particular, see Definition 3.3 in~\cite{kempehittingtimes}) to this model, we obtain a natural generalization of the classical definition of hitting time, and a well-known formula for its computation can be derived for the quantum case with analogous methods.

The classical setting is as follows: we start from a state $X_0$ distributed according to a state probability vector $q$, and we want to reach a given state $z \in S$. We want to quantify how many steps, in expectation, we need to run before reaching such state.
\begin{definition}
    \label{def:classical-ht}
    Given a Markov chain $(S, P)$, a starting state $q$ and a target state $z \in S$, define the following random variable:
    \begin{align*}
        T_z(q) = \min\{ t \ge 0 : X_t = z, X_0 \sim q \} .
    \end{align*}
    The $z$-hitting time is defined as the expectation of such variable:
    \begin{align*}
        h_z(q) := \Expect{T_z(q)} .
    \end{align*}
\end{definition}
\noindent A more complete approach on the classical theory of hitting times can be also found in~\cite{norrismarkovchains}.

For the quantum case, we can extend a definition proposed by Kempe for unitary walks to our model. In particular, given a quantum Markov chain $(\calH, \calE)$ and a target state $\ket{z}$, we consider a process in which we alternate applications of the evolution map $\calE$ with a measurement of the observable
\begin{align*}
    O_z = \Pi_z - \Pi_{-z} .
\end{align*}
where $\Pi_z := \ketbra{z}{z}$ and $\Pi_{-z} := \id - \ketbra{z}{z}$. If such measurement returns $+1$, then the state will collapse to our target state, and the process will stop. Otherwise, we repeat the procedure by executing $\calE$ and trying another measurement.
\begin{definition}
    \label{def:quantum-ht}
    Let $(\calH, \calE)$ be a quantum Markov chain, and let $\ket{z} \in \calH$ be an arbitrary state. Considering the procedure described above, let $M_i$ be a random variable equal to $1$ if and only if the $i$-th measurement makes the state collapse to $\ket{z}$, the quantum hitting time $h_z(\rho)$ is defined as the expectation of
    \begin{align*}
        T_z(\rho) = \min\{ t \ge 0 \,|\, M_t = 1 \}
    \end{align*}
    \revdtext{where $\rho$ is the initial state of the chain.}
\end{definition}
One can see that this definition is closely related the classical one: in fact, both Definitions~\ref{def:classical-ht} and~\ref{def:quantum-ht} imply a check on the state to see whether $z$ was reached or not, and this bit of information is the only one we use in order to decide whether to stop the algorithm. However, in the quantum case some quantum coherences are in general destroyed at each step. It is worth highlighting here that such measurement is nevertheless different from a full collapse of the quantum state, which would always return a probabilistic mixture of measurement eigenstates, hence yielding a classical Markov process over the elements of the measurement basis itself (exactly what we did in Eq.~(\ref{eq:quantum-to-classical-construction})). \revdtext{We show a way to carry out such partial maeasurement using an ancilla qubit in Section~\ref{sec:qiskit-implementation}.}

We now derive a formula for computing such hitting times in the quantum general case. Consider the projector map $\calP_{-z}(\rho) := \Pi_{-z} \rho \Pi_{-z}$, which basically removes from a state $\rho$ the probability mass in the subspace spanned by our target state $\ket{z}$. Physically, $\calP_{-z}$ represents the action of measuring with $O_z$ and halting when its outcome is $+1$. For the rest of the work we define $\rho_{-z} := \calP_{-z}(\rho)$ and $\calE_{-z} := \calP_{-z} \circ \calE \circ \calP_{-z}$.
\begin{theorem}
    \label{thm:quantum-ht-formula}
    Let $(\calH, \calE)$ be a quantum Markov chain and fix a target state $\ket{z}$. Moreover, suppose that each eigenvalue $\lambda$ of the map $\calE_{-z}$ satisfies $|\lambda| < 1$. The following holds:
    \begin{align*}
        h_z(\rho) = \tr\left[ (\calI - \calE_{-z})^{-1} (\rho_{-z}) \right] .
    \end{align*}
    \revdtext{where $\calI$ is the identity map.}
\end{theorem}
Before proving Theorem~\ref{thm:quantum-ht-formula}, let us give an interpretation of the results. \revdtext{First of all, notice that any eigenvalue of $\calE$ (without projections) associated to an \emph{Hermitian} eigenstate must be exactly $1$, as the map is trace-preserving. As an example, if $\calE(\rho) = U \rho U^\dag$, then $\calE(\ket{\phi_i}\bra{\phi_j}) = \lambda_i \lambda_j^* \ket{\phi_i}\bra{\phi_j}$, where $\ket{\phi_k}$ is the eigenstate of $U$ associated with the eigenvalue $\lambda_k$. On the other hand, eigenvalues associated to non-diagonal eigenoperators can be unitary (as in the case we have just shown), or even $0$ (as in the example of Eq.~(\ref{eq:quantum-to-classical-construction})).} The requirement of Theorem~\ref{thm:quantum-ht-formula} on the eigenvalues of $\calE_{-z}$ has a precise meaning: under the assumption $|\lambda| < 1$ for the eigenvalues of $\calE_{-z}$, each eigenvalue of $(\calI - \calE_{-z})$ is non-zero, and thus the map is invertible, implying that the hitting time is finite. On the other hand, a unitary eigenvalue would suggest that the corresponding eigenstate $\rho$ is completely outside of the target subspace $\vspan{\ket{z}}$. As a consequence, for any $k$, also $\calE^k(\rho)$ will always have zero overlap with $\ket{z}$. Under these conditions, each measurement would give zero probability of collapsing to the target state, and the hitting time would be infinite.

A second observation is that this result is in full correspondence with the known classical counterpart: indeed, for a classical Markov chain $(S, P)$ the formula for the hitting time as in Definition~\ref{def:classical-ht} (as used, e.g., by Magniez et al.~\cite{magniezhitting}) is:
\begin{align*}
    h_z(q) = q_{-z} (\id - P_{-z})^{-1} \vv{1} .
\end{align*}
where $P_{-z}$ (resp.\ $q_{-z}$) is taken from $P$ (resp.\ $q$) by zeroing out the columns/rows corresponding to $z$, and $\vv{1}$ is the vector of all ones. This suggests that our formalization of the quantum hitting time is indeed a natural extension of the classical one. We would also like to remark that the target state $\ket{z}$ can be replaced with an arbitrary target subspace $S \subseteq \calH$: in this case, the process will stop as soon as the state collapses to \emph{any} of the states within $S$. This can be achieved by replacing $\Pi_z$ with the projector onto $S$.

\begin{proof}[Proof of Theorem~\ref{thm:quantum-ht-formula}]
    We are going to use the law of total expectation on the outcome of the first measurement. Let $E$ be the event occurring when the first measurement makes the state collapse to $\ket{z}$. Conditioning on such event we have:
    \begin{align*}
        h_z(\rho) = \Expect{T_z(\rho) \given E} \Prob{E} + \Expect{T_z(\rho) \given \bar{E}} \Prob{\bar{E}} ,
    \end{align*}
    where $\Prob{E} = \tr(\Pi_z \rho)$ is the probability of measuring $\ket{z}$ from the starting state $\rho$ using the observable $O_z$. When $E$ occurs, we measure our target state in zero steps, implying that
    \begin{align*}
        \Expect{T_z(\rho) \given E} = 0 .
    \end{align*}
    If $E$ does not occur, then the state will collapse to the post-measurement state
    \begin{align*}
        \rho \mapsto \frac{\Pi_{-z} \rho \Pi_{-z}}{\tr(\Pi_{-z} \rho)} = \frac{\calP_{-z}(\rho)}{\tr(\calP_{-z}(\rho))} .
    \end{align*}
    Then the algorithm will apply the step procedure $\calE$ once to this state, and the process will start over, by the 
    Markov property. Therefore, the second conditional expectation can be expressed as:
    \begin{align*}
        \Expect{T_z(\rho) \given \bar{E}} = 1 + h_z\left( \calE\bigg( \frac{\calP_{-z}(\rho)}{\tr(\calP_{-z}(\rho))} \bigg) \right) .
    \end{align*}
    In order to simplify the proof from now on, we make an educated guess: we assume that $h_z(\rho)$ is linear in $\rho$. Notice that this only adds constraints to our functional equation, therefore, any solution we find is also valid in the original setting \footnote{This constraint is actually not necessary to conclude the proof. However, assuming that the function is linear significantly simplifies the exposition.}. The equation above can be rewritten as:
    \begin{align*}
        h_z(\rho) & = \Pr[\bar{E}] + h_z(\calE \circ \calP_{-z} (\rho)) = \tr(\calP_{-z} (\rho)) + h_z(\calE \circ \calP_{-z} (\rho))
    \end{align*}
    and, by developing this recurrence $N$ times,  we obtain that:
    \begin{align*}
        h_z(\rho) & = \sum_{k = 0}^N \tr(\calP_{-z} \circ (\calE \circ \calP_{-z})^k (\rho)) + h_z((\calE \circ \calP_{-z})^{N+1} (\rho)) .
    \end{align*}
    By exploiting the fact that $\calP^2_{-z} \equiv \calP_{-z}$, the second term becomes, for $N \ge 2$:
    \begin{align*}
        h_z((\calE \circ \calP_{-z})^{N+1} (\rho)) & = h_z\bigg( \calE \circ (\calP_{-z} \circ \calE \circ \calP_{-z})^N (\rho) \bigg) = h_z\bigg( \calE \circ \calE_{-z}^N (\rho) \bigg) .
    \end{align*}
    Taking the limit as $N \rightarrow \infty$, $\calE^N_{-z} \rightarrow \vv{0}$ (by the condition on its eigenvalues), and thus the last term tends to $h_z(0) = 0$, by linearity of $h_z$. Therefore, we find that the hitting time satisfies
    \begin{align*}
        h_z(\rho) & = \sum_{k = 0}^\infty \tr(\calP_{-z} \circ (\calE \circ \calP_{-z})^k (\rho)) \\
        & = \sum_{k = 0}^\infty \tr((\calP_{-z} \circ \calE \circ \calP_{-z})^k \circ \calP_{-z} (\rho)) \\
        & = \sum_{k = 0}^\infty \tr(\calE_{-z}^k(\rho_{-z})) = \tr\left[\sum_{k = 0}^\infty \calE_{-z}^k(\rho_{-z})\right] .
    \end{align*}
    A geometric sum argument concludes the proof.
\end{proof}
\revdtext{Notice that the last expression essentially exploits the well-known fact that, for a non-negative integer random variable, $\Expect{X} = \sum_{k = 0}^\infty \Prob{X \ge k}$. Thus, if we think of the hitting time as an average over the possible directed paths connecting initial and target states, then the $k$-th term $\tr(\calE_{-z}^k(\rho))$ represents a contribution from all the paths of length $\ge k$.}

%% file: contents/generalized-ht.tex
\section{Generalized Hitting Time}
\label{sec:generalized-hitting-times}
Definition~\ref{def:quantum-ht} states that a measurement with the observable $O_z$ must be carried out at each step. As this procedure removes relevant quantum coherences \revdtext{at each step, this may lead to undesired effects, such as a quantum Zeno effect.} Thus, we further generalize the idea of hitting time as follows: instead of checking at each step, \revdtext{we extract a number of steps $T$ from a fixed discrete probability distribution $\sigma$ over the natural numbers, e.g., using a classical algorithm that samples from $\sigma$. Notice that we will indicate with $A \sim \beta$ the fact that a random variable $A$ is distributed according to the probability distribution $\beta$.} We then run the step procedure $\calE$ for $T$ times before \revdtext{attempting a measurement}. Figure~\ref{fig:generalized-ht-algorithm} gives a  visual representation of the process.

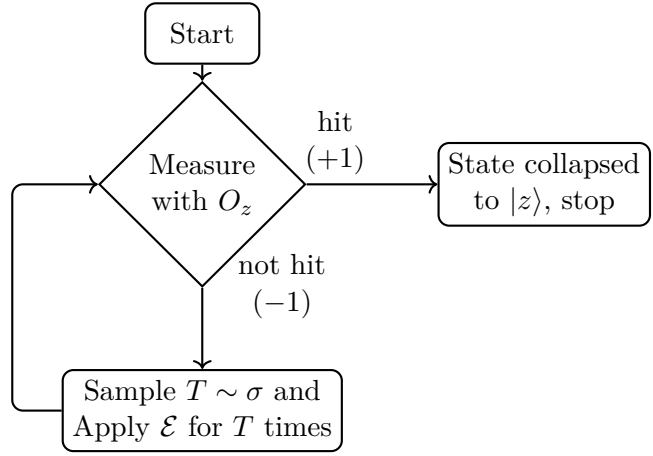
\begin{figure}
    \centering
    \input{figures/generalized-ht-algorithm}
    \caption{\revdtext{Flow chart summarizing the algorithm exploiting the generalized hitting time. The generalized hitting time $h^\sigma_z(\rho)$ gives the (expected) number of executions of $\calE$ before the measurement gives our desired state $\ket{z}$, starting from the initial state $\rho$. Notice that, in the case of Definition~\ref{def:quantum-ht}, $T \equiv 1$ so $\calE$ is called exactly once per iteration.}}
    \label{fig:generalized-ht-algorithm}
\end{figure}

This approach allows us to apply the classical analysis of hitting times to quantum processes relying on prolonged coherent dynamics, such as Grover's algorithm~\cite{groversalgorithm} (see Section~\ref{sec:grover-algorithm}).
\begin{definition}
    \label{def:generalized-ht}
    Let $\{ X_k \}_{k \in \N}$ be an infinite sequence of independent random variables where $X_k \sim \sigma$. Moreover, define $Q$ as the number of times we measure with $O_z$ without hitting $z$, before succeeding. The generalized $z$-hitting time $h^\sigma_z(\rho)$ with respect to $\sigma$ is the expectation of the following random variable
    \begin{align*}
        T_z^\sigma(\rho) = \sum_{k = 1}^Q X_k .
    \end{align*}
\end{definition}
\noindent \revdtext{Here, $X_k$ represents the number of steps we execute after measuring the $k$-th time with failure. This means that we will execute $\calE^{X_k}$ before attempting the next measurement. Definition~\ref{def:quantum-ht} can be re-obtained by choosing $\sigma$ to be the trivial distribution where $X_k = 1$ with probability~$1$.}
Notice that this generalized definition is only relevant in the quantum case. Indeed, since measurements do not affect the behaviour of a classical walk, not looking at the state at some point of the evolution can only increase the hitting time. We remark that $Q$ can also be $0$, in the case where the first measurement succeeds (we do a measurement at the beginning, like in the basic case). Let us also define the map
\begin{align*}
    \calE^\sigma := \Expect[T \sim \sigma]{\calE^T} = \sum_{t = 0}^\infty \Prob[\sigma]{T = t} \calE^t
\end{align*}
This quantum operation represents the process of sampling a $T \sim \sigma$ and running $\calE$ for $T$ times. We then have the following result.
\begin{theorem}
    \label{thm:generalized-ht-formula}
    Let $(\calH, \calE)$ be a quantum Markov chain and fix a target state $\ket{z}$, and suppose that each eigenvalue $\lambda$ of the map $\calE^\sigma_{-z}$ satisfies $|\lambda| < 1$. The following holds:
    \begin{align*}
        h^\sigma_z(\rho) = \Expect[\sigma]{T} \cdot \tr\left[ (\calI - \calE^\sigma_{-z})^{-1} (\rho_{-z}) \right]
    \end{align*}
\end{theorem}
Informally, notice that, if we consider the Markov chain $(\calH, \calE^\sigma)$, the trace operation in the formula of Theorem~\ref{thm:generalized-ht-formula} returns exactly the hitting time of such chain, i.e., the number of times $\calE^\sigma$ is executed before hitting $z$. Then the expectation of $T$ represents the number of executions of $\calE$ per execution of $\calE^\sigma$. Despite this result seems reasonable, a formal argument is required to show that statistical interdependence between $Q$ and $X_k$'s does not influence the hitting time.
\begin{proof}[Proof of Theorem~\ref{thm:generalized-ht-formula}]
    We follow a similar approach as in Theorem~\ref{thm:quantum-ht-formula}. Considering the event $E$ occurring when the first measurement of $O_z$ returns $+1$ and conditioning the expectation on this event we obtain, as in Theorem~\ref{thm:quantum-ht-formula}:
    \begin{align*}
        h_z^\sigma(\rho) & = \Expect{T^\sigma_z(\rho) \given \bar{E}} \Prob{\bar{E}} .
    \end{align*}
    What essentially changes here is the computation of this conditional expectation, since in this case we will apply $\calE^{X_1}$, for $X_1 \sim \sigma$. Hence, let us condition on the value of this variable:
    \begin{align*}
        \Expect{T^\sigma_z(\rho) \given \bar{E}} & = \sum_{t = 0}^\infty \Expect{T^\sigma_z(\rho) \given X_1 = t, \bar{E}} \cdot \Prob[\sigma]{X_1 = t} .
    \end{align*}
    If the measurement does not hit $\ket{z}$, and $t$ steps of $\calE$ are executed, the conditional expectation will be:
    \begin{align*}
        \Expect{T^\sigma_z(\rho) \given X_1 = t, \bar{E}} & = t + h^\sigma_z\left( \calE^t\bigg( \frac{\calP_{-z}(\rho)}{\tr(\calP_{-z}(\rho))} \bigg) \right) .
    \end{align*}
    This because the process will start over, by the Markov property, after applying $\calE$ for $t$ times to the post-measurement state, and we incur a cost of $t$ in terms of hitting times to do so. As in Theorem~\ref{thm:quantum-ht-formula}, we make the educated guess that $h^\sigma_z$ is linear in $\rho$, and this implies that:
    \begin{align*}
        h^\sigma_z(\rho) & = \Prob{\bar{E}} \cdot \sum_{t = 0}^\infty \left[t + h^\sigma_z\left( \calE^t\bigg( \frac{\calP_{-z}(\rho)}{\tr(\calP_{-z}(\rho))} \bigg) \right) \right] \cdot \Prob[\sigma]{X_1 = t} .
    \end{align*}
    Since $\Prob{\bar{E}} = \tr(\calP_{-z}(\rho))$, we can exploit linearity to simplify:
    \begin{align*}
        h^\sigma_z(\rho) & = \sum_{t = 0}^\infty \left[t \cdot \Prob{\bar{E}} + h^\sigma_z( \calE^t \circ \calP_{-z}(\rho)) \right] \cdot \Prob[\sigma]{X_1 = t} .
    \end{align*}
    By splitting the sum, we obtain that the first term is $\Prob{E} \cdot \Expect[\sigma]{T}$ by definition of expectation. Again by linearity, the second term becomes:
    \begin{align*}
        \sum_{t = 0}^\infty h^\sigma_z( \calE^t \circ \calP_{-z}(\rho)) & = h^\sigma_z\left[ \bigg( \sum_{t = 0}^\infty \calE^t \cdot \Prob[\sigma]{X_1 = t} \bigg) \circ \calP_{-z}(\rho) \right] \\
        & = h^\sigma_z ( \calE^\sigma \circ \calP_{-z}(\rho) ) .
    \end{align*}
    To sum up, the recurrence relation for $h^\sigma_z$ is:
    \begin{align*}
        h^\sigma_z(\rho) & = \tr(\calP_{-z}(\rho)) \cdot \Expect[\sigma]{T} + h^\sigma_z ( \calE^\sigma \circ \calP_{-z}(\rho) ).
    \end{align*}
    By developing this recurrence as we did in the proof of Theorem~\ref{thm:quantum-ht-formula}, we obtain
    \begin{align*}
        h^\sigma_z(\rho) & = \Expect[\sigma]{T} \cdot \tr\left[\sum_{k = 0}^\infty (\calE^\sigma)_{-z}^k(\rho_{-z})\right].
    \end{align*}
    A geometric sum argument concludes the proof.
\end{proof}
\revdtext{Also in this case, we can decompose the last expression as a sum over contributions of the trajectories}
\begin{align*}
        h^\sigma_z(\rho) & = \sum_{k = 0}^\infty \Expect[\sigma]{T} \cdot \tr\left[(\calE^\sigma)_{-z}^k(\rho_{-z})\right].
\end{align*}
\revdtext{If we remove $\Expect[\sigma]{T}$ we obtain Definition~\ref{def:quantum-ht} with respect to the chain $(\calH, \calE^\sigma)$. The slight difference is that the paths considered in the $k$-th term of this sum are not of length $k$ anymore. Instead, these trajectories have a length $k T$, where $T$ is distributed according to $\sigma$. The expectation of $T$ takes this random length into account.}

%% file: figures/generalized-ht-algorithm.tex
\begin{tikzpicture}[scale=0.5]
    \begin{scope}[every node/.style={minimum height=0.8cm,minimum width=1.5cm}]
        \node[draw,rectangle,rounded corners,thick] (st) at (0,0) {Start};
        \node[draw,diamond,thick,align=center] (meas) at (0,-4) {Measure\\with $O_z$};
        \node[draw,rectangle,rounded corners,thick,align=center] (smp) at (0,-10) {Sample $T \sim \sigma$ and\\Apply $\calE$ for $T$ times};
        \node[draw,rectangle,rounded corners,thick,align=center] (stop) at (9,-4) {State collapsed\\to $\ket{z}$, stop};

        \node[align=center] at (2.1, -6.7) {not hit\\($-1$)};
        \node[align=center] at (3.5, -2.9) {hit\\($+1$)};

        \draw[->, thick] (st) -- (meas);
        \draw[->, thick] (meas) -- (smp);
        \draw[->, thick, rounded corners] (smp) -- (-5, -10) -- (-5,-4) -- (meas);

        \draw[->, thick, rounded corners] (meas) -- (stop);
    \end{scope}
\end{tikzpicture}

%% file: contents/grover.tex
\section{Hitting Times for Grover's Algorithm}
\label{sec:grover-algorithm}
In this section we show how to apply our analysis of hitting times to a well-known problem in the literature: the \emph{unstructured search problem}. Defining $N = 2^n$, the problem consists of a function $f : [N] \rightarrow \{0, 1\}$ given as an oracle, for which there is exactly one value $x_0 \in [N]$ such that $f(x_0) = 1$. Any classical algorithm needs $\Omega(N)$ queries in order to output $x_0$ with non-trivial success probability. However, a famous algorithm by Grover~\cite{groversalgorithm, groversalgorithm2} can find such value within $\bigO(\sqrt{N})$ queries to the oracle.

Roughly speaking, Grover's algorithm starts in the state $\ket{+} := \ket{+}^{\otimes n}$. Using the oracle of $f$, one can construct a unitary $G$ that rotates the state in the plane $\calH^* = \vspan{ \ket{+}, \ket{x_0} }$ by an angle of
\begin{align*}
    2\gamma = 2\arcsin{\frac{1}{\sqrt{N}}} .
\end{align*}
Grover~\cite{groversalgorithm} showed that, after applying $G$ for $r = \Theta(\sqrt{N})$ times, the success probability can be shown to be:
\begin{align*}
    \big|\bra{x_0} G^r \ket{+} \big|^2 = 1 - \bigO\left(\frac{1}{N}\right)
\end{align*}
The original algorithm runs $G$ for $r$ times and then tries a full measurement in the computational basis, making the state collapse to the desired state with the aforementioned probability. Here, we want to use the hitting time framework to recover the same result. As a warm up, let us consider the case where we measure at each step. Therefore, we are considering a quantum Markov chain $(\calH, \calG)$ where $\calH = \vspan{ \ket{x} : x \in \N }$ and
\begin{align*}
    \calG(\rho) = G \rho G^\dag ,
\end{align*}
and we want to compute the hitting time $h_{x_0}(\proj{+})$ with respect to this chain.

\smallskip
\begin{theorem}
    \label{thm:grover-ht-1}
    $h_{x_0}(\proj{+}) = N/4$.
\end{theorem}
This result basically shows that measuring at each step nullifies all the quantum speed-up given by Grover's algorithm. It is important to remark, however, that this is \emph{not} a classical walk, as we are not doing a full measurement in the computational basis at each step. Indeed, the $\frac{1}{4}$ factor is still not achievable by a classical walk.
\begin{proof}[Proof of Theorem~\ref{thm:grover-ht-1}]
    First of all, notice that both $G$ and $O_z$ are invariant with respect to $\calH^*$, i.e.\ any state in $\calH^*$ will stay in such subspace after an application of the unitary $G$ or a measurement of $O_z$. Since these are the only two operations we are going to apply, we can restrict ourselves to this two-dimensional subspace. Hence, we can analyze the hitting time with respect to the quantum Markov chain $(\calH^*, \calG|_{\calH^*})$, where $\calG|_{\calH^*}$ is the application of the unitary $G$ restricted to $\calH^*$:
    \begin{align*}
        G|_{\calH^*} = 
        \begin{bmatrix}
            \cos 2\gamma & -\sin 2\gamma \\
            \sin 2\gamma & \cos 2\gamma
        \end{bmatrix}
    \end{align*}
    We can represent $\calE, \calI, \calP_{-x_0}$ as $4 \times 4$ matrices, and density matrices as $4$-dimensional vectors, we can apply the formula of Theorem~\ref{thm:quantum-ht-formula} (the trace operator becomes the dot product with the representation of the identity matrix). By doing this computations with a software of symbolic calculation~\cite{Mathematica}, one can check that this formula gives~$\frac{N}{4}$.
\end{proof}
As expected, destroying quantum coherence too often hinders the performance of Grover's algorithm. As a second attempt, we consider the case in which the measurement is carried out at each step only with probability $p$. The analysis of such processes coincides with the notion of the generalized hitting time (Definition~\ref{def:generalized-ht}), where $\sigma$ is the geometric distribution with parameter $p$.

To make use of the formula given by Theorem~\ref{thm:generalized-ht-formula}, we first need to compute a closed form representation for the linear map $\calE^\sigma$. In the case of a geometric distribution, this can be achieved easily as detailed in the following lemma.

\begin{lemma}
    \label{thm:fm-channel-geometric-formula}
    If $\sigma = \mathrm{Geom}(p)$, then
    \begin{align*}
        \calE^\sigma = p \calE \circ (\calI - (1 - p)\calE)^{-1}
    \end{align*}
\end{lemma}
Note that such operation is always well-defined as long as the probability $p$ of measuring is non-zero: since $\calE$ is completely positive and trace-preserving, any eigenvalue $\lambda$ of $\calE$ satisfies $|\lambda| \le 1$, and thus all the eigenvalues $\mu$ of the map $\calI - (1 - p) \calE$ satisfy $|\mu| \ge p > 0$, giving us an invertible map.
\begin{proof}[Proof of Lemma~\ref{thm:fm-channel-geometric-formula}]
    Since $\sigma$ is a geometric distribution, $\Prob[\sigma]{T = t} = p(1 - p)^{t-1}$. Therefore, by definition of $\calE^\sigma$:
    \begin{align*}
        \calE^\sigma & = \sum_{t = 1}^\infty p(1 - p)^{t-1} \calE^t \\
        & = p \calE \circ \sum_{t = 1}^\infty (1 - p)^{t-1} \calE^{t-1} \\
        & = p \calE \circ \sum_{t = 0}^\infty \big[ (1 - p) \calE \big]^t \\
        & = p \calE \circ (\calI - (1 - p) \calE)^{-1}
    \end{align*}
\end{proof}
We can now prove the following general result.
\begin{theorem}
    \label{thm:grover-ht-2}
    Let $(\calH, \calG)$ be a quantum Markov chain describing a Grover search as above, and let $p > 0$ be the probability of measuring at each step. The generalized z-hitting time of such process is
    \begin{align*}
        h_{x_0}^\sigma(\proj{+}) = \Theta\bigg( Np + \frac{1}{p} \bigg)
    \end{align*}
\end{theorem}
As a first observation, this claim extends the result of Theorem~\ref{thm:grover-ht-1}, proving that no constant probability $p$ can give quantum speed-up in Grover's search. Secondly, our upper bound is fully consistent with the well-known result about the optimality of Grover's algorithm, proving a lower bound of $\Omega(\sqrt{N})$ quantum queries for the unstructured search problem. Indeed, our upper bound reaches its optimum when we choose $p = \Theta(N^{-1/2})$, achieving an optimal algorithm. \revdtext{A possible implementation of the algorithm on digital quantum processors is shown in Section~\ref{sec:qiskit-implementation}.}
\begin{proof}[Proof of Theorem~\ref{thm:grover-ht-2}]
    As in the proof of Theorem~\ref{thm:grover-ht-1}, we can simply study the two-dimensional quantum Markov chain $(\calH^*, \calG|_{\calH^*})$, since we are still applying the same set of operations to the starting state $\proj{+}$. Again, we represent every super-operator as $4 \times 4$ matrix and quantum states as $4$-dimensional vectors. What changes here is the formula we have to compute: we first use Lemma~\ref{thm:fm-channel-geometric-formula} to compute the matrix representation of $\calE^\sigma$ and then use this to apply Theorem~\ref{def:generalized-ht} (for a geometric distribution, $\Expect[\sigma]{T} = \frac{1}{p}$). With the help of a software for symbolic calculation~\cite{Mathematica}, we obtain that:
    \begin{align*}
        h_{x_0}(\proj{+}) & = \frac{1}{p} \frac{N^2 p^2 - 16Np + 16N + 16p - 16}{8N - 4Np}
    \end{align*}
    For $p = 1$ we obtain $\frac{N}{4}$, as entailed by Theorem~\ref{thm:grover-ht-1}. For $p = o(1)$, this is asymptotically equal to $\frac{Np}{8} + \frac{2}{p}$.
\end{proof}

%% file: contents/cycle.tex
\section{Traversing the cycle}
\label{sec:hitting-times-cycle}

\begin{figure}
    \begin{subfigure}{0.49\textwidth}
        \centering
        \input{figures/even-cycle}
    \end{subfigure}
    \begin{subfigure}{0.49\textwidth}
        \centering
        \input{figures/odd-cycle}
    \end{subfigure}
    
    \caption{Quantum walk on a $2N$-cycle (left), and on a $(2N+1)$-cycle (right). We start from the state $\ket{0}$ (with coin register on $\ket{{\uparrow}}$, and we want to reach the state $\ket{N}$.}
    \label{fig:quantum-walk-even-cycle}
\end{figure}
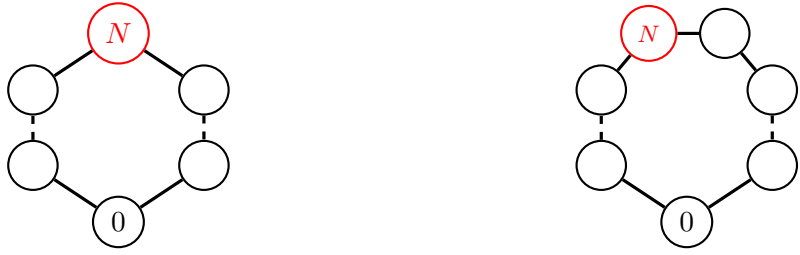

In this section we consider a different example, showing that results of Section~\ref{sec:hitting-times} can be exploited not only to derive analytical expressions for hitting times, but also to provide numerical estimations. Consider a cyclic structure with $2N$ sites such as the one depicted in  Fig.~\ref{fig:quantum-walk-even-cycle} (left) and a unitary Markov chain $(\calH, \calU)$ whose Hilbert space can be decomposed as $\calH = \calH_C \otimes \calH_S$ where $\calH_S = \vspan{ \ket{x} : x \in \Z_{2N} }$ denotes the nodes along the cycle, and $\calH_C = \vspan{ \ket{{\uparrow}}, \ket{{\downarrow}} }$ denotes a coin register. This setup is known in the literature as a \emph{coined unitary walk}~\cite{kempeintro}, where the coin register is used to move to different sites in superposition. In our case, $\calU$ applies a unitary of the form $U = S (H_C \otimes \id_S)$, where $H$ is the Hadamard gate, and $S$ moves from $\ket{x}$ to $\ket{x+1}$ if the coin is in state $\ket{{\uparrow}}$, and to $\ket{x-1}$ if the coin is in state $\ket{{\downarrow}}$. Therefore, the unitary $U$ acts on the total state as follows:
\begin{align*}
    U \ket{{\uparrow}, x} & = \frac{\ket{{\uparrow}, x+1} + \ket{{\downarrow}, x-1}}{\sqrt{2}} \\
    U \ket{{\downarrow}, x} & = \frac{\ket{{\uparrow}, x+1} - \ket{{\downarrow}, x-1}}{\sqrt{2}}
\end{align*}
We start from the state $\ket{{\uparrow}, 0}$ and we want to reach the opposite state in the cycle, namely the state $\ket{N}$ (regardless of the state of the coin register). We can apply the formula from Theorem~\ref{def:quantum-ht}, where the projector outside the subspace of $\ket{N}$ is
\begin{align*}
    \Pi_{-N} = \id_C \otimes (\id_S - \proj{N})
\end{align*}
Using \emph{numpy}~\cite{harrisnumpy}, we can numerically compute the matrices for $\calU, \calP_{-N}$, and the formula can be calculated using matrix algebra, just like we did in Section~\ref{sec:grover-algorithm}, taking $\sigma = Geom(p)$ and $p = \frac{1}{100}, \frac{2}{100}, \ldots, 1$. \revdtext{The results of these computations are summarized in the plots of Fig.~\ref{fig:cycle-numerical-plots}. We can observe some interesting behaviours}: in the case where $p = 1$ (which coincides with the regular hitting time), numerical estimations give an hitting time scaling as $N^2$, which is consistent with the symmetric classical case (see Gambler's Ruin~\cite{norrismarkovchains}). On the other hand, the optimal choice for $p$ seems to depend on the parity of $N$: while an odd $N$ always yields a quadratic hitting time, for even values of $N$ there exist a $p$ that achieves a hitting time scaling linearly with the size of the cycle. We repeated these computations for a cycle of odd length (see Fig.~\ref{fig:quantum-walk-even-cycle}, right) and, unlike the even case, under optimal $p$ any value of $N$ led to a linear hitting time.

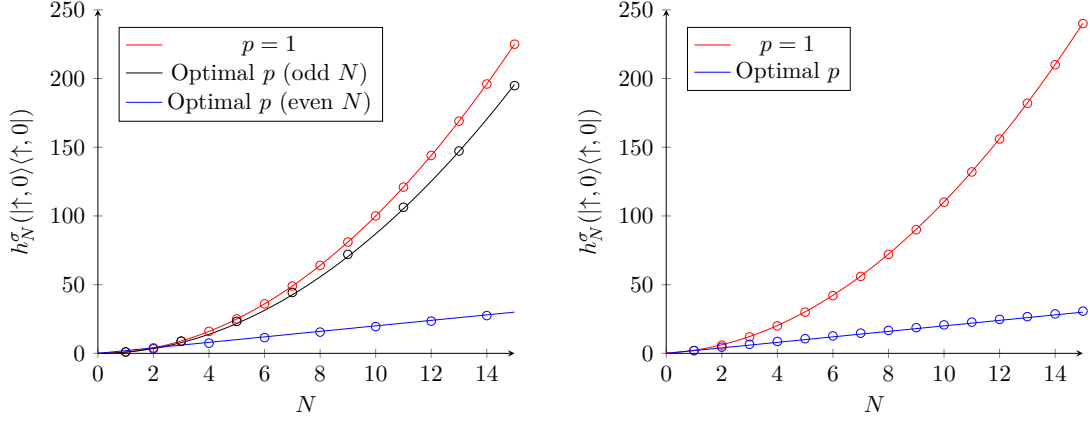
\begin{figure*}
    \begin{subfigure}{0.49\textwidth}
    \centering
        \input{figures/even-cycle-plot}
    \end{subfigure}
    \begin{subfigure}{0.49\textwidth}
        \centering
        \input{figures/odd-cycle-plot}
    \end{subfigure}

    \caption{Plots of the hitting time for the opposite side of a cycle of length $2N$ (left) and of length $2N + 1$ (right) with respect to $N$. The circles are the numerical computations of the formula given by Theorem~\ref{thm:generalized-ht-formula}, while the continuous lines are the regressions. In the left plot, when $p = 1$ the regression gives $N^2$, while for the optimal $p$ we obtain $\sim 2N$ (even $N$) and $\sim N^2/1.15$ (odd $N$). In the right case, any $N$ yields a linear hitting time with the optimal $p$.}
    \label{fig:cycle-numerical-plots}
\end{figure*}

%% file: figures/even-cycle.tex
\begin{tikzpicture}[scale=0.5]
    \begin{scope}[every node/.style={circle,thick,draw,minimum size=0.65cm}]
        \node (s0) at (0,0) {$0$};
        \node (s1) at (-2.25,1.5) {};
        \node (sNm) at (-2.25,3.5) {};
        \node[red] (sN) at (0,5) {$N$};
        \node (sNp) at (2.25,3.5) {};
        \node (sL) at (2.25,1.5) {};
    \end{scope}

    \begin{scope}[every node/.style={fill=white},
                every edge/.style={draw=black,very thick}]
        \path [-] (s0) edge (s1);
        \path [-, dashed] (s1) edge (sNm);
        \path [-] (sNm) edge (sN);
        \path [-] (sN) edge (sNp);
        \path [-, dashed] (sNp) edge (sL);
        \path [-] (sL) edge (s0);
    \end{scope}
\end{tikzpicture}

%% file: figures/odd-cycle.tex
\begin{tikzpicture}[scale=0.5]
    \begin{scope}[every node/.style={circle,thick,draw,minimum size=0.65cm}]
        \node (s0) at (0,0) {$0$};
        \node (s1) at (-2.25,1.5) {};
        \node (sNm) at (-2.25,3.5) {};
        \node[red] (sN) at (-1,5) {\footnotesize $N$};
        \node (sN1) at (1,5) {};
        \node (sNp) at (2.25,3.5) {};
        \node (sL) at (2.25,1.5) {};
    \end{scope}

    \begin{scope}[every node/.style={fill=white},
                every edge/.style={draw=black,very thick}]
        \path [-] (s0) edge (s1);
        \path [-, dashed] (s1) edge (sNm);
        \path [-] (sNm) edge (sN);
        \path [-] (sN) edge (sN1);
        \path [-] (sN1) edge (sNp);
        \path [-, dashed] (sNp) edge (sL);
        \path [-] (sL) edge (s0);
    \end{scope}
\end{tikzpicture}

%% file: figures/even-cycle-plot.tex
\begin{tikzpicture}[scale=0.8]
    \begin{axis}[
        legend style={at={(0.05,0.95)},anchor=north west},
        axis lines = left,
        ymin=0,
        ymax=250,
        xlabel = \(N\),
        ylabel = {\(h^\sigma_{N}(\proj{{\uparrow}, 0})\)},
    ]
            
    \addplot[domain=0:15, samples=100, color=red] {x^2};
    \addlegendentry{$p = 1$}
    \addplot[domain=0:15, samples=100, color=black] {x^2/1.15};
    \addlegendentry{Optimal $p$ (odd $N$)} \addplot[domain=0:15, samples=100, color=blue] {2*x};
    \addlegendentry{Optimal $p$ (even $N$)}
    
    \addplot[only marks, color=red, mark=o]
        coordinates {
            (1, 1.0)
            (2, 3.9999999999999973)
            (3, 8.999999999999984)
            (4, 15.999999999999936)
            (5, 24.999999999999833)
            (6, 35.99999999999961)
            (7, 48.99999999999913)
            (8, 63.99999999999814)
            (9, 80.99999999999639)
            (10, 99.99999999999429)
            (11, 120.99999999999017)
            (12, 143.99999999998454)
            (13, 168.99999999997794)
            (14, 195.99999999996783)
            (15, 224.9999999999553)
        };
            
    \addplot[only marks, color=blue, mark=o]
        coordinates {
            (2, 3.313758865248225)
            (4, 7.425664739884359)
            (6, 11.453770491803013)
            (8, 15.466096256683983)
            (10, 19.47368421052515)
            (12, 23.479790575914098)
            (14, 27.48207253885456)
        };
        
    \addplot[only marks, color=black, mark=o]
        coordinates {
            (1, 1.0)
            (3, 8.784782608695632)
            (5, 23.246031746031555)
            (7, 44.3121374045796)
            (9, 72.00067669172626)
            (11, 106.31296296295525)
            (13, 147.25176470586007)
            (15, 194.81934306565472)
        };
    \end{axis}
\end{tikzpicture}

%% file: figures/odd-cycle-plot.tex
\begin{tikzpicture}[scale=0.8]
    \begin{axis}[
        legend style={at={(0.05,0.95)},anchor=north west},
        axis lines = left,
        ymin=0,
        ymax=250,
        xlabel = \(N\),
        ylabel = {\(h^\sigma_{N}(\proj{{\uparrow}, 0})\)},
    ]
            
    \addplot[domain=0:15, samples=100, color=red] {{((2*x+1)/2)^2}};
    \addlegendentry{$p = 1$}
    \addplot[domain=0:15, samples=100, color=blue] {2*x};
    \addlegendentry{Optimal $p$}
    
    \addplot[only marks, color=red, mark=o]
        coordinates {
            (1, 1.9999999999999996)
            (2, 5.999999999999993)
            (3, 11.999999999999972)
            (4, 19.9999999999999)
            (5, 29.999999999999705)
            (6, 41.999999999999375)
            (7, 55.99999999999857)
            (8, 71.99999999999747)
            (9, 89.99999999999514)
            (10, 109.99999999999281)
            (11, 131.99999999998835)
            (12, 155.9999999999829)
            (13, 181.9999999999722)
            (14, 209.99999999996226)
            (15, 239.9999999999456)
        };
    
    \addplot[only marks, color=blue, mark=o]
        coordinates {
            (1, 1.9999999999999996)
            (2, 4.495208173971728)
            (3, 6.590476190476152)
            (4, 8.633447363659812)
            (5, 10.6579590068791)
            (6, 12.673993371240405)
            (7, 14.685538284397916)
            (8, 16.694122465785885)
            (9, 18.700045100359876)
            (10, 20.70607376000837)
            (11, 22.70984924961037)
            (12, 24.71490644843236)
            (13, 26.71568881795163)
            (14, 28.71973377703356)
            (15, 30.725488470839263)
        };
    \end{axis}
\end{tikzpicture}

%% file: contents/qiskit.tex
\section{Implementation with digital quantum circuits}
\label{sec:qiskit-implementation}

\begin{figure}
    \centering
    \input{figures/qiskit-circuit}
    \caption{Circuit implementing a step with measurement for the generalized hitting time. \revdtext{$R_X$ rotates the state of the ancilla so that the probability of measuring $1$ is exactly $p$. $b_i$ is the bit recording whether the target state was hit at step $i$. The walk is stopped whenever $\calE(\rho_{i-1})$ is in the target state $\ket{111}$ and qubit $R$ is found in state $\ket{1}$ (such that the measurement is performed).}}
    \label{fig:qiskit-circuit}
\end{figure}
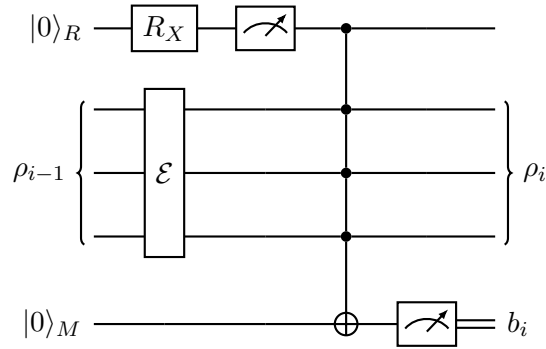

The generalized hitting-time paradigm as shown in Section~\ref{sec:generalized-hitting-times} consists of interrupting the execution of a quantum circuit as soon as (1) the measurement takes place (this happens with probability $p$), and (2) the observable for the hit yields a $+1$, making the state collapse to the target subspace.

Although this paradigm seems conceptually simple, such a hybrid process is not straightforward to implement on current quantum computing platforms. In fact, one essentially needs a classical process that interacts with the quantum device in an \emph{online} fashion, conditioning the execution on the measurement outcomes. In principle, this represents a possible intriguing use case for dynamic circuits~\cite{Corcoles2021,Cross_2022}.

In this section we show how to carry out hitting time experiments in the standard unitary quantum circuit model using Qiskit~\cite{qiskit}. To be completely general, we assume that is not possible to conditionally interrupt the execution of a quantum circuit, hence we fix a number $N$ of steps to run and, besides the $n$ qubits needed to store the state of our walk, we also allocate $N$ classical bits $\{ b_i \}_i$, where $b_i$ stores the measurement outcome at the $i$-th step. Two extra qubits are used to emulate the classical coin flip which decides whether to try the measurement or not ($q_R$), and to carry out the measurement ($q_M$). The quantum circuit implementing a single step is depicted in Figure~\ref{fig:qiskit-circuit}. Here we use $R_X$ to rotate $q_R$ by an angle $\theta = \arcsin{\sqrt{p}}$, so that the measurement yields $\ket{1}$ with probability $p$. The measurement qubit $q_M$ will be flipped if $q_R$ is in state $\ket{1}$ and $\rho$ is in the state $\ket{111}$, which is the target state in our toy example. One could change the target state by simply replacing the \revdtext{$4$-way Toffoli gate $T$} with $U \cdot T \cdot U^\dagger$, where $U$ is a unitary mapping $\ket{111}$ onto a different vector. After the measurement, $b_i$ will tell whether the walk should stop at the $i$-th step.
\revdtext{The circuit of Figure~\ref{fig:qiskit-circuit} transforms the state as follows:}
\begin{align*}
    \rho_i & = \tr_R\big[ T \cdot R_X (\proj{00}_{RM} \otimes \calE(\rho_{i-1})) R_X^\dag \cdot T^\dag \big] \\
    & = (1 - p) \calE(\rho_{i-1}) \otimes \proj{0}_M \\
    & + p [\calP_z(\calE(\rho_{i-1})) \otimes \proj{0}_M
    + \calP_{-z}(\calE(\rho_{i-1})) \otimes \proj{1}_M]
\end{align*}
Running this circuit $N$ times will yield $N$ classical bits, from $b_1$ to $b_N$, and the hitting time can be computed as the smallest $i$ such that $b_i = 1$ \footnote{To be consistent with Definition~\ref{def:generalized-ht}, a measurement has to be done also at the beginning, yielding $b_0$, so that the experimental hitting time can be also zero like its theoretical counterpart.}. We used this construction to run Grover's algorithm as described in Section~\ref{sec:grover-algorithm}. All numerical simulation confirmed the predictions of Theorem~\ref{thm:grover-ht-2}: \revdtext{Figure~\ref{fig:qiskit-plot} shows four plots for different choices of $N$, comparing the average hitting time of $1000$ executions and the theoretical expectation computed analytically.}

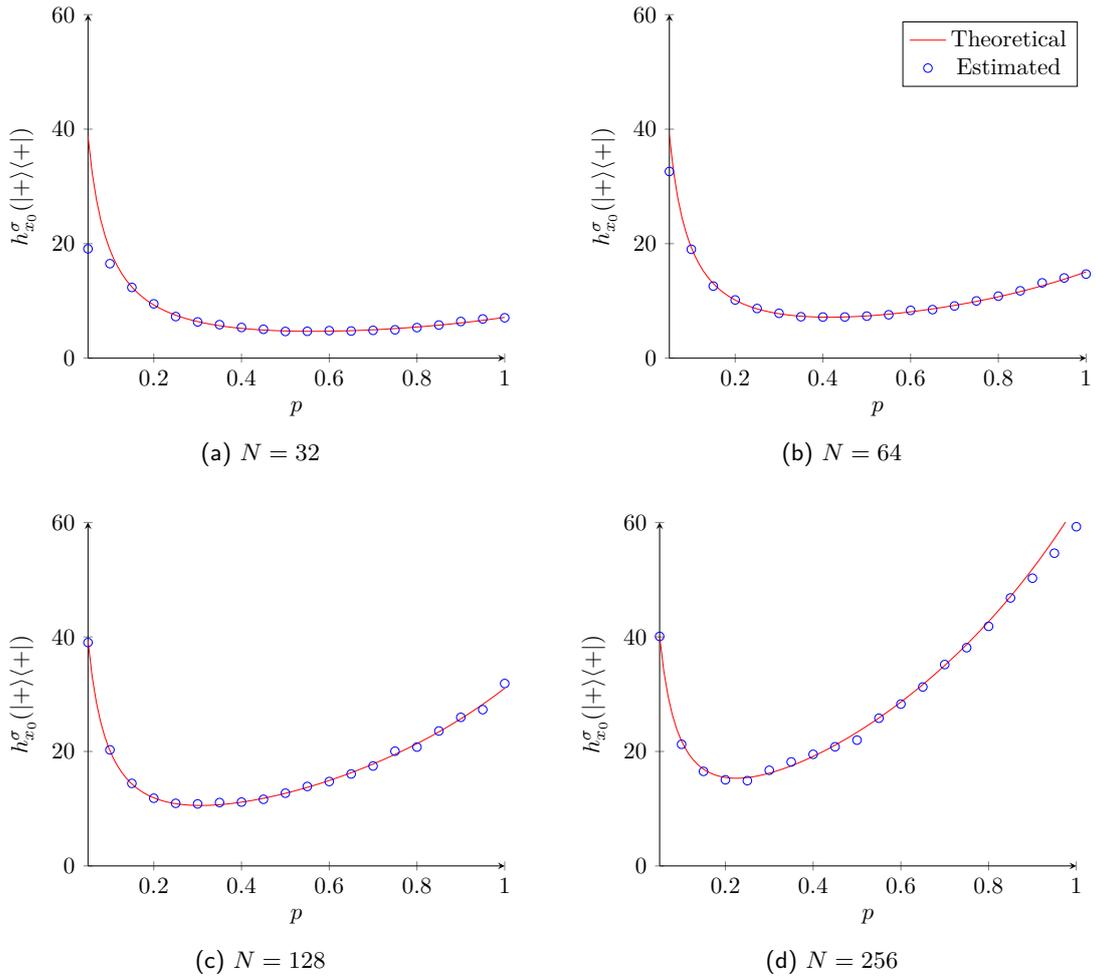
\begin{figure*}
    \centering
    \input{figures/qiskit-plot}
    \caption{Plots of the concurrent hitting time with respect to the step-wise measurement probability $p$. The red line is the theoretical expectation as computed in the proof of Theorem~\ref{thm:grover-ht-2}, while the blue circles are the estimations computed by running the circuit with the given $p$ for 1000 independent times.}
    \label{fig:qiskit-plot}
\end{figure*}
\bigskip

%% file: figures/qiskit-circuit.tex
\begin{quantikz}
    \lstick{$\ket{0}_R$} & \gate{R_X} & \meter{} & \ctrl{1} & \qw & \qw \\
    \lstick[wires=3]{$\rho_{i-1}$}
    & \gate[wires=3]{\calE} & \qw & \ctrl{1} & \qw & \qw\rstick[wires=3]{$\rho_i$} \\
    & & \qw & \ctrl{1} & \qw & \qw \\
    & & \qw & \ctrl{1} & \qw & \qw \\
    \lstick{$\ket{0}_M$} & \qw & \qw & \targ{} & \meter{} & \cw\rstick{$b_i$}
\end{quantikz}

%% file: figures/qiskit-plot.tex
\begin{subfigure}{0.49\textwidth}
    \centering
    \begin{tikzpicture}[scale=0.8]
        \begin{axis}[
            axis lines = left,
            ymin=0,
            ymax=60,
            xlabel = \(p\),
            ylabel = {\(h^\sigma_{x_0}(\proj{+})\)},
        ]
            
        \addplot [
            domain=0.05:1, 
            samples=100, 
            color=red,
        ] {1/x * ((32 * x)^2 + 16*32 - 20*32*x + 16*x) / (8*32 - 4*32*x)};
            
        \addplot [
            only marks,
            color=blue,
            mark=o,
        ]
        coordinates {
            (0.05, 19.094)(0.1, 16.502)(0.15, 12.35)(0.2, 9.487)(0.25, 7.261)(0.3, 6.327) (0.35, 5.85)(0.4, 5.361)(0.45, 5.039)(0.5, 4.665)(0.55, 4.691)(0.6, 4.805)(0.65, 4.737)(0.7, 4.854)(0.75, 4.949)(0.8, 5.338)(0.85, 5.783)(0.9, 6.402)(0.95, 6.847)(1, 7.049)
            };
        \end{axis}
    \end{tikzpicture}
    \caption{$N = 32$}
\end{subfigure}
\hfill
\begin{subfigure}{0.49\textwidth}
    \centering
    \begin{tikzpicture}[scale=0.8]
        \begin{axis}[
            axis lines = left,
            ymin=0,
            ymax=60,
            xlabel = \(p\),
            ylabel = {\(h^\sigma_{x_0}(\proj{+})\)},
        ]
            
        \addplot [
            domain=0.05:1, 
            samples=100, 
            color=red,
        ]
        {1/x * ((64 * x)^2 + 16*64 - 20*64*x + 16*x) / (8*64 - 4*64*x)};
        \addlegendentry{Theoretical}
            
        \addplot [
            only marks,
            color=blue,
            mark=o,
        ]
        coordinates {
            (0.05, 32.605)(0.1, 19.011)(0.15, 12.59)(0.2, 10.166)(0.25, 8.67)(0.3, 7.833) (0.35, 7.233)(0.4, 7.182)(0.45, 7.19)(0.5, 7.341)(0.55, 7.58)(0.6, 8.346)(0.65, 8.477)(0.7, 9.115)(0.75, 9.994)(0.8, 10.829)(0.85, 11.764)(0.9, 13.148)(0.95, 14.008)(1, 14.684)
        };
        \addlegendentry{Estimated}
        \end{axis}
    \end{tikzpicture}
    \caption{$N = 64$}
\end{subfigure}
\vskip\baselineskip
\begin{subfigure}{0.49\textwidth}
    \centering
    \begin{tikzpicture}[scale=0.8]
        \begin{axis}[
            axis lines = left,
            ymin=0,
            ymax=60,
            xlabel = \(p\),
            ylabel = {\(h^\sigma_{x_0}(\proj{+})\)},
        ]
            
        \addplot [
            domain=0.05:1, 
            samples=100, 
            color=red,
        ]
        {1/x * ((128 * x)^2 + 16*128 - 20*128*x + 16*x) / (8*128 - 4*128*x)};
            
        \addplot [
            only marks,
            color=blue,
            mark=o,
        ]
        coordinates {
            (0.05, 39.025)(0.1, 20.276)(0.15, 14.422)(0.2, 11.83)(0.25, 10.939)(0.3, 10.834) (0.35, 11.075)(0.4, 11.155)(0.45, 11.643)(0.5, 12.711)(0.55, 13.885)(0.6, 14.747)(0.65, 16.066)(0.7, 17.461)(0.75, 20.041)(0.8, 20.764)(0.85, 23.575)(0.9, 25.958)(0.95, 27.301)(1, 31.869)
        };
    \end{axis}
    \end{tikzpicture}
    \caption{$N = 128$}
\end{subfigure}
\hfill
\begin{subfigure}{0.49\textwidth}
    \centering
    \begin{tikzpicture}[scale=0.8]
        \begin{axis}[
            axis lines = left,
            ymin=0,
            ymax=60,
            xlabel = \(p\),
            ylabel = {\(h^\sigma_{x_0}(\proj{+})\)},
        ]
            
        \addplot [
            domain=0.05:1, 
            samples=100, 
            color=red,
        ]
        {1/x * ((256 * x)^2 + 16*256 - 20*256*x + 16*x) / (8*256 - 4*256*x)};
            
        \addplot [
            only marks,
            color=blue,
            mark=o,
        ]
        coordinates {
            (0.05, 40.096)(0.1, 21.247)(0.15, 16.519)(0.2, 15.047)(0.25, 14.899)(0.3, 16.709) (0.35, 18.171)(0.4, 19.495)(0.45, 20.815)(0.5, 21.975)(0.55, 25.805)(0.6, 28.269)(0.65, 31.265)(0.7, 35.176)(0.75, 38.124)(0.8, 41.849)(0.85, 46.814)(0.9, 50.262)(0.95, 54.625)(1, 59.248)
        };
        \end{axis}
    \end{tikzpicture}
        
    \caption{$N = 256$}
\end{subfigure}

%% file: contents/conclusions.tex
\section{Conclusions}

In this work we \revdtext{described and developed} a general way to formalize quantum random walks, including both coherent dynamics and classical randomness. Within this framework, the notion of concurrent hitting time, originally given only for unitary walks~\cite{kempehittingtimes}, can be naturally extended \revdtext{using the interpretation of quantum measurements as CP maps, and an explicit formula -- directly related to a well-known classical counterpart -- can be derived. Such results were used to study some paradigmatic examples of quantum walks, including Grover's algorithm, both analytically and numerically. A few concluding remarks are in order.}

\revdtext{First, it is worth observing that, for an for a $N$-state walk, computing the hitting time through Theorem~\ref{thm:quantum-ht-formula} and its generalisations requires the manipulation of $N^2 \times N^2$ matrices -- representing the related CP maps -- thus requiring $\bigO(N^6)$ arithmetic operations with a naive approach. Possible future works could hence investigate more efficient numerical approaches with the aim of treating larger problem instances.}

\revdtext{Concerning our analysis of Grover's algorithm, it is important to notice that, while our formulation is in principle fully equivalent to its traditional form~\cite{groversalgorithm, groversalgorithm2}, the process of Figure~\ref{fig:generalized-ht-algorithm} never discards the quantum state upon failure but always restarts from the post-measurement one, i.e.\ only one copy of the initial state is necessary. This feature could be useful in practical implementations and, more in general, for extending the search algorithm to an oblivious amplitude amplification scheme.}

\revdtext{Finally, in the setting of the cycle we showed that reaching the opposite site takes a quadratic time if we use the basic protocol where measurements are taken at each step, just like in a classical balanced walk. Instead, the situation changes if we use the protocol associated to the  generalized hitting time with a geometric distribution of parameter $p$. In particular, for an odd cycle (i.e.\ of length $2N + 1$), one can always find a suitable choice for $p$ such that the opposite side is hit in linear time. On the contrary, for an even cycle of length $2N$ no $p$ achieves better than quadratic time when $N$ is odd. These results closely resemble what Kendon and Tregenna report in Ref.~\cite[Sec.~V]{kendontregennanumerical} for mixing times. This is not surprising if one recalls that a tight connection between hitting and mixing times exists already in classical setting~\cite{levinMarkovChainsMixing2017}.}

\revdtext{Overall, our results suggest that a complete theory of quantum walks and of the tools required to analyze their behaviour could be developed further by leveraging some tight analogies with classical Markov chains. Our formula for generalised concurrent hitting times is, in fact, just the tip of the iceberg, compared to the rich toolbox known in the literature to estimate and bound classical hitting times. As possible next steps, it could for example be interesting to extend our analysis towards drift theory and the study of random decline~\cite{he2004drift, doerrmultdrift}.}